\documentclass[12pt]{iopart}

\usepackage{amsthm,graphicx}
\usepackage{amssymb}
\def\softd{{\leavevmode\setbox1=\hbox{d}%
\hbox to 1.05\wd1{d\kern-0.4ex{\char039}\hss}}}
\def\softt{{\leavevmode\setbox1=\hbox{t}%
\hbox to \wd1{t\kern-0.6ex{\char039}\hss}}}

\newcommand{\N}{\mathbb{N}}
\newcommand{\R}{\mathbb{R}}
\newcommand{\Z}{\mathbb{Z}}
\newcommand{\D}{\mathrm{d}}
\newcommand{\ee}{\mathrm{e}}
\newcommand{\imag}{\mathrm{i}}
\newcommand{\pd}{\partial}
\newcommand{\HH}{\mathcal{H}}

\newtheorem{theorem}{Theorem}[section]
\newtheorem{proposition}{Proposition}[section]

\renewcommand{\theequation}{\arabic{section}.\arabic{equation}}

\begin{document}

\title[Spectrum of Dirichlet Laplacian in a conical layer]
{Spectrum of Dirichlet Laplacian in a conical layer
}

\author{Pavel Exner and Milo\v{s} Tater}
\address{Doppler Institute for Mathematical Physics and Applied
Mathematics, \\ B\v{r}ehov\'{a} 7, 11519 Prague, \\ and  Nuclear
Physics Institute ASCR, 25068 \v{R}e\v{z} near Prague, Czechia}
\ead{exner@ujf.cas.cz, tater@ujf.cas.cz}

\begin{abstract}
We study spectral properties of Dirichlet Laplacian on the conical
layer of the opening angle $\pi-2\theta$ and thickness equal to
$\pi$. We demonstrate that below the continuum threshold which is
equal to one there is an infinite sequence of isolated eigenvalues
and analyze properties of these geometrically induced bound
states. By numerical computation we find examples of the eigenfunctions.
\end{abstract}

\maketitle

\section{Introduction} \label{s: intro}

Geometrically induced bound states in infinitely stretched regions
with a hard-wall (or Dirichlet) boundary attracted a lot attention
in recent years, in particular, because they represent a purely
quantum effect which cannot be approached from the usual
semiclassical point of view. The attention was first paid to tubes
where the effective attraction is caused by bending \cite{ES89,
GJ92, DE95}, later the analogous question was addressed with one
dimension added, in curved layers \cite{DEK01, CEK04}. This
question is more complicated and to the date we do not have a full
analogue of the result saying that \emph{any} local bending of an
asymptotically straight tube exhibits such bound states. One
reason is that geometry of surfaces is more complicated than that
of curves and global properties play a more important role.

However, this is not the only difference between tubes and layers.
In the case of the former we have various solvable examples. They
often correspond to tubes with a non-smooth boundary being thus
beyond the assumptions under which the general results are proved
and illustrating the robustness of the effect; on the other hand
they can be composed of simple building blocks which facilitates
their solvability. In fact, some examples came already before the
general theory \cite{LLM86, ESS89, SRW89} and others were found in
parallel with it --- see, e.g., \cite{LCM99} for a review.

This is in contrast with our state of knowledge concerning the
layers. We have reasonably strong and general existence results
\cite{DEK01, CEK04}, later even extended to layers in higher
dimensions \cite{LL07}, as well a particular ``weak coupling''
result \cite{EK01}, but to the date no solvable example is
known\footnote{By solvability we do not mean finding the
eigenvalues and eigenfunctions analytically, of course, this is
beyond reach even in the tubes case. What we have in mind is the
possibility of bringing the analysis to the state when one can
find these quantities numerically and study their dependence on
the parameters of the model.}. The goal of the present work is to
discuss such an example in the geometrical situation which is
particularly simple --- apart of a trivial scaling the model
depends on a single parameter.

The layer to consider is the region bordered by two conical
surfaces shifted mutually along the cone axis by $\pi \csc\theta$. Using a natural separation of variables the angular component can be separated and the problem can be in each partial wave reduced to analysis of an appropriate operator on a skewed semi-infinite strip. This will be demonstrated in the next section and used in Sec.~3 to derive some properties of the corresponding discrete spectrum. To illustrate these results and to provide further insights into this binding mechanism, we present in the following section a numerical analysis of the problem: we calculate the eigenvalue dependence on the cone opening angle and present examples of eigenfunctions.

\section{The Hamiltonian}\label{s: hamiltonian}

\subsection{Basic definitions}
\label{ss: operator}

As we have said the object of our study is the Hamiltonian of a
quantum particle confined to a hard-wall conical layer of the
opening angle $\pi-2\theta$ for a fixed $\theta\in (0,
\frac12\pi)$, i.e.
 \begin{equation} \label{conlayer}
 \Sigma_\theta := \{ (r,\phi,z)\in\R^3:\: (r,z)\in
 \Omega_\theta,\, \phi\in [0,2\pi)\}\,,
 \end{equation}
where
 \begin{equation} \label{consection}
 \Omega_\theta := \{ (r,z)\in\R^2:\: r\ge 0,\, 0 < z-r
 \sin\theta < \pi \csc\theta\}\,.
 \end{equation}

\begin{figure}
   \begin{center}
     \includegraphics[width=27em]{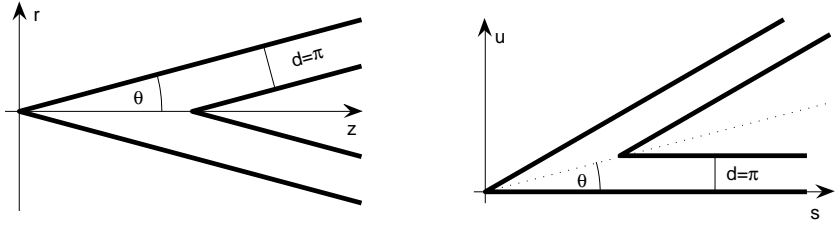}
     \caption{The problem geometry. The left part shows the
setting in cylindrical coordinates as described in Secs.~\ref{ss:
operator} and \ref{ss: partial}, the right part shows the cross section in rotated
coordinates $s,u$ used in Sec.~\ref{ss: alternative}.}\label{fig1}
   \end{center}
\end{figure}

\noindent Without loss of generality we may assume that the layer width is
$\pi$, the result for other widths being obtained by a trivial
scaling. Similarly we neglect values of physical constants and
identify the Hamiltonian $H^\theta$ with the Dirichlet Laplacian
$-\Delta_\mathrm{D}^{\Sigma_\theta}$ on the Hilbert space
$\HH=L^2(\Sigma_\theta,\, r\D r\D\phi\D z)$. Using the cylindrical
coordinates we can write the operator action explicitly,
 \begin{equation} \label{ham1}
 H^\theta\psi = -\frac{1}{r} \frac{\pd}{\pd r}\left(r \frac{\pd\psi}{\pd
 r}\right) - \frac{1}{r^2} \frac{\pd^2\psi}{\pd\phi^2}
 - \frac{\pd^2\psi}{\pd z^2}\,,
 \end{equation}
with the domain consisting of all functions $\psi\in W^{1,2}
_\mathrm{loc} (\Sigma_\theta)$ such that $\psi(\vec x)=0$ for
$\vec x= (r,\phi,z)\in \pd\Sigma_\theta$ and $H\psi\in L^2$. The operator is associated with the quadratic form
 \begin{equation} \label{hamform1}
 Q^\theta[\psi] = \int_{\Sigma_\theta} \left[ \left| \frac{\pd\psi}{\pd
 r} \right|^2 + \frac{1}{r^2} \left| \frac{\pd\psi}{\pd
 \phi} \right|^2 + \left| \frac{\pd\psi}{\pd z} \right|^2
 \right](r,\phi,z)\, r\D r\D\phi\D z
 \end{equation}
defined on functions $\psi\in W^{1,1} _\mathrm{loc}
(\Sigma_\theta)$ vanishing at the layer boundary and such that
$Q^\theta[\psi] < \infty$. It is well known
\cite[Sec.~XIII.15]{RS78} that the operator $H$ is self-adjoint
and positive, $\sigma(H) \subset [0,\infty)$.

\subsection{Partial wave decomposition} \label{ss: partial}

Using the cylindrical symmetry of the system we write state
Hilbert space as
 \begin{equation} \label{partial1}
 \HH = L^2(\Omega_\theta,\, r\D r\D z) \otimes L^2(S_1) =
 \bigoplus_{m\in\Z} L^2(\Omega_\theta,\, r\D r\D z)\,,
 \end{equation}
where $L^2(S_1)$ refers conventionally to functions on the unit
circle with the orthonormal basis $\{ \frac{1}{2\pi} \ee^{\imag
m\phi}:\: m\in\Z\}$. The operator $H^\theta$ decomposes
accordingly as
 \begin{equation} \label{partial2}
 H^\theta = \bigoplus_{m\in\Z} H^\theta_m\,, \quad H^\theta_m\psi
 = -\frac{1}{r} \frac{\pd}{\pd r}\left(r \frac{\pd\psi}{\pd r}\right)
 - \frac{\pd^2\psi}{\pd z^2} + \frac{m^2}{r^2}\psi \,,
 \end{equation}
with the domain consisting of $\psi\in W^{1,2} _\mathrm{loc}
(\Omega_\theta)$ satisfying the requirement $\psi(r,
r\sin\theta)=\psi(r, r\sin\theta+ \pi\tan\theta)=0$ and such that
$H_m\psi\in L^2$. Under these conditions the operators
$H^\theta_m,\, m\ne 0,$ are self-adjoint while for $m=0$ one has
to demand additionally that $\left.\frac{\pd\psi}{\pd r}
\right|_{r=0} =0$. All the boundary values here are understood as
the appropriate limits. The quadratic forms associated with the
partial wave components of the Hamiltonian are
 \begin{equation} \label{partialform1}
 Q^\theta_m[\psi] = \int_{\Sigma_\theta} \left[ \left| \frac{\pd\psi}{\pd
 r} \right|^2 + \left| \frac{\pd\psi}{\pd z} \right|^2 + \frac{m^2}{r^2}
 |\psi|^2 \right](r,z)\, r\D r\D z
 \end{equation}
being defined on all $\psi\in W^{1,1} _\mathrm{loc}
(\Omega_\theta)$ satisfying $\psi(r, r\sin\theta)=\psi(r,
r\sin\theta+ \pi\csc\theta)=0$ and such that $Q^\theta_m[\psi] <
\infty$.

As usual we can pass to unitarily equivalent operators by means of
the trans\-formation $U:\: L^2(\Omega_\theta,\, r\D r\D z) \to
L^2(\Omega_\theta)$ defined by $(U\psi)(r,z):= r^{1/2} \psi(r,z)
\equiv \tilde\psi(r,z)$ and an analogous map $L^2(\Sigma_\theta,\,
r\D r\D\phi\D z) \to L^2(\Sigma_\theta)$. In particular, the
operator $\tilde H^\theta_m:= UH^\theta_m U^{-1}$ acts as
 \begin{equation} \label{partial3}
 \tilde H^\theta_m\tilde\psi = - \frac{\pd^2\tilde\psi}{\pd r^2}
 - \frac{\pd^2\tilde\psi}{\pd z^2} + \frac{1}{r^2}
 \left(m^2-\frac14 \right) \tilde\psi
 \end{equation}
on functions $\tilde\psi\in W^{1,2}_\mathrm{loc} (\Omega_\theta)$
such that $\tilde H^\theta_m\tilde\psi\in L^2$ satisfying the
Dirichlet boundary condition at the boundary of $\Omega_\theta$;
the associated quadratic form is
 \begin{equation} \label{partialform2}
 \tilde Q^\theta_m[\tilde\psi] = \int_{\Sigma_\theta}
 \left[ \left| \frac{\pd\tilde\psi}{\pd r} \right|^2 +
 \left| \frac{\pd\tilde\psi}{\pd z} \right|^2 + \frac{1}{r^2}
 \left(m^2-\frac14 \right) |\tilde\psi|^2 \right](r,z)\, \D r\D z
 \end{equation}
with the domain of all $\tilde\psi\in W^{1,1} _\mathrm{loc}
(\Omega_\theta)$ satisfying the Dirichlet condition at
$\pd\Omega_\theta$ and such that $Q^\theta_m[\tilde\psi] <
\infty$.

\subsection{Alternative expressions} \label{ss: alternative}

There are other ways how to express action of our Hamiltonian. For
instance, one can use rotated coordinates in the radial cross
section plane according to Fig.~\ref{fig1},
 \begin{equation} \label{coord}
 s = r\cos\theta + z\sin\theta\,, \quad u = -r\sin\theta + z\cos\theta\,,
 \end{equation}
in which we have
 \begin{equation} \label{consection2}
 \Omega_\theta := \{ (s,u)\in\R^2:\: s\ge 0,\, 0 < u <
 \min( \pi, s\cot\theta) \}\,.
 \end{equation}
The partial wave operators act now as
 \begin{equation} \label{partial4}
 \tilde H^\theta_m\tilde\psi = - \frac{\pd^2\tilde\psi}{\pd s^2}
 - \frac{\pd^2\tilde\psi}{\pd u^2}
 + \frac{4m^2-1}{4(s\cos\theta - u\sin\theta)^2}\, \tilde\psi\,,
 \end{equation}
the domain being the same as above; the associated quadratic form
takes in these coordinates the shape
 \begin{equation} \label{partialform3}
 \hspace{-2em} \tilde Q^\theta_m[\tilde\psi] = \int_{\Sigma_\theta}
 \left[ \left| \frac{\pd\tilde\psi}{\pd s} \right|^2 +
 \left| \frac{\pd\tilde\psi}{\pd u} \right|^2
 + \frac{4m^2-1}{4(s\cos\theta - u\sin\theta)^2}
 |\tilde\psi|^2 \right](s,u)\, \D s\D u\,.
 \end{equation}
One can also use skew coordinates,
 \begin{equation} \label{skewcoord}
 y = s -u\tan\theta\,, \quad v = u\,,
 \end{equation}
in which the expression (\ref{partialform3}) acquires the form
 \begin{eqnarray}
 \hspace{-2em} \tilde Q^\theta_m[\tilde\psi] &\!=\!& \int_0^\infty \D y
 \int_0^\pi \D v \left[ \frac{1}{\cos^2\theta} \left| \frac{\pd\tilde\psi}{\pd y} \right|^2 +
 \left| \frac{\pd\tilde\psi}{\pd v} \right|^2
 + 2\tan\theta \: \mathrm{Re} \left( \frac{\pd\bar{\tilde\psi}}{\pd y}
 \frac{\pd\tilde\psi}{\pd v} \right)\right. \nonumber \\ && \hspace{-2em}
 \left. \phantom{\left|\frac{\pd\tilde\psi}{\pd r} \right|^2} + \frac{4m^2-1}{4
 y^2\cos^2\theta}\,
 |\tilde\psi|^2 \right](y,v)\,.
 \label{partialform4}
 \end{eqnarray}

\section{Properties of the spectrum}\label{s: properties}

Since the conical layer is built over a surface with a cylindrical
end its basic spectral properties follow from the theory developed
on \cite{DEK01, CEK04}.
\begin{theorem} \label{disc_spec}
For any fixed $\theta\in(0,\frac12\pi)$ we have
$\sigma_\mathrm{ess}(H^\theta)= [1,\infty)$ while the discrete spectrum
of the operator is non-empty with $\dim \sigma_\mathrm{disc}
(H^\theta) = \infty$.
\end{theorem}
\begin{proof}
For most part of the argument we can refer to the papers quoted
above. In particular, we proved there that $\inf
\sigma_\mathrm{ess} (H^\theta)= 1$. In order to check that any
$k^2+1$ belongs to the spectrum one has to construct a suitable
Weyl sequence. Consider three dimensional Cartesian coordinates
$s,t,u$ and take the family of functions
 $$
 g_n(s,t,u):= \sqrt{\frac{2n}{n-2}}\: \ee^{\imag ks} f(ns,nt)\,
 \cos \frac{\pi nu}{n-2}
 $$
where $|u|\le \frac{n-2}{n}$ and $f\in C_0^\infty(\R^2)$ has the
support in $(-1,1)^2$. By construction $\|(-\Delta-k^2-1)g_n\|\to
0$ as $n\to\infty$, and since the vanishing quantity is invariant
w.r.t. Euclidean transformations we can find a family of functions
$\tilde g_n$ supported by rectangular blocks of the size $2n
\times 2n \times (1-\frac 2n)$, each of them contained in
$\Sigma_\theta$ being naturally further and further from the layer
tip as $n$ grows, which proves the claim.

The existence of the discrete spectrum was proven in \cite{DEK01,
CEK04}. While it is not stated explicitly there that the number of
the eigenvalues in $(0,1)$ is infinite, it is clear from the proof
in which one constructs a sequence of trial functions supported by
mutually disjoint ring-shaped domains of the conical layer, and
thus orthogonal to each other; for large enough indices they give
the value of the energy functional below the essential spectrum
threshold.
\end{proof}

After demonstrating the existence of the discrete spectrum let us
look at its properties in more detail. First we note that all the
geometrically induced bound states in the layer belong to the
$s$-wave sector.
\begin{proposition} \label{s-states}
$\sigma_\mathrm{disc} (H^\theta_m) = \emptyset$ holds if $m\ne 0$.
\end{proposition}
\begin{proof}
The claim is equivalent to $\sigma_\mathrm{disc} (\tilde
H^\theta_m) = \emptyset$. For any higher partial wave, $m\ne 0$, we have $\tilde Q^\theta_m[\tilde \psi] \ge \int_{\Omega_\theta}
|\vec\nabla \tilde\psi|^2(r,z)\, \D r\D z$ which means that
$\tilde H^\theta_m \ge -\Delta_\mathrm{D} ^{\Omega_\theta}$ in the form sense where the \emph{rhs} acts according to \cite[Sec.~XIII.15]{RS78} on functions $\psi\in W^{1,2}
_\mathrm{loc} (\Omega_\theta)$ such that $\psi(\vec x)=0$ for
$\vec x= (r,z)\in \pd\Omega_\theta$ and $-\Delta\psi\in L^2$. The result then follows by minimax principle because
$\sigma(-\Delta_\mathrm{D} ^{\Omega_\theta}) =
\sigma_\mathrm{ess}(-\Delta_\mathrm{D} ^{\Omega_\theta}) =
[1,\infty)$.
\end{proof}

This does not mean, of course, that the cylindrically asymmetric
part of the problem is trivial --- it can exhibit, for instance,
an interesting scattering behaviour. Since we are dealing with the
discrete spectrum here, however, we consider from now on only
functions independent of the angular variable $\phi$ and drop the
index $m$ in the quadratic form expressions. By
Theorem~\ref{disc_spec} the operator $H^\theta$ has a sequence of
eigenvalues,
 \begin{equation} \label{eigenval}
 \lambda_1(\theta) < \lambda_2(\theta) \le \lambda_3(\theta) \le
 \cdots\,,
 \end{equation}
labelled in the ascending order; the corresponding eigenfunctions
will be denoted as $\psi^{(j)}\in L^2(\Sigma_\theta),\:
j=1,2,\dots,$; with the symmetry in mind we use the same symbols
for their radial cuts $\psi^{(j)}\in L^2(\Omega_\theta,\, r\D r\D
z)$ and $\tilde\psi^{(j)}$ for their counterparts in
$L^2(\Omega_\theta)$.

The main question is about the dependence of the spectrum on the
cone opening angle. Let us start with demonstrating that the
eigenvalues (\ref{eigenval}) are smooth functions of $\theta$. To
this aim we compare the cross sections (\ref{consection2}) for two
different values of the angle, $\theta$ and $\theta'$. They can be
mapped onto each other by a longitudinal scaling, $s\cot\theta =
s'\cot\theta'$; then we can express the form (\ref{partialform3})
referring to the angle $\theta'$ as
 \begin{eqnarray}
 \tilde Q^{\theta'}[\tilde\psi] &\!=\!& \int_{\Sigma_\theta}
 \left[ \frac{\tan\theta}{\tan\theta'}
 \left|\frac{\pd\tilde\psi}{\pd s} \right|^2 +
 \frac{\tan\theta'}{\tan\theta}
 \left| \frac{\pd\tilde\psi}{\pd u} \right|^2 \right. \nonumber \\ &&
 \left. \phantom{\left|\frac{\pd\tilde\psi}{\pd r} \right|^2}
 \hspace{-2em} - \frac{\sin 2\theta}{\sin 2\theta'}\,
 \frac{|\tilde\psi|^2}{4(s\cos\theta - u\sin\theta)^2} \right](s,u)\,
 \D s\D u\,. \label{partialform3'}
 \end{eqnarray}
The integration is thus taken over the same domain and the angle
dependence is moved to the coefficients of the operator.

\begin{proposition}
The functions $\lambda_j(\cdot),\: j=1,2,\dots,$ are $C^1$ on
$(0,\frac12\pi)$.
\end{proposition}
\begin{proof}
To find the derivative of $\lambda_j(\cdot)$ we have
Feynman-Hellmann theorem differentiate (\ref{partialform3'}) with
respect to $\theta'$ at $\theta'=\theta$. It yields
 $$
 \lambda'_j(\theta) = \int_{\Sigma_\theta}
 \left[ -\frac{2}{\sin 2\theta}
 \left|\frac{\pd\tilde\psi^{(j)}}{\pd s} \right|^2 +
 \frac{2}{\sin 2\theta} \left| \frac{\pd\tilde\psi^{(j)}}{\pd u} \right|^2
 + \frac{|\tilde\psi^{(j)}|^2\cot 2\theta}{2(s\cos\theta - u\sin\theta)^2}
 \right](s,u)\,
 \D s\D u\,.
 $$
Since $\tilde\psi^{(j)}\in W^{1,1}_\mathrm{loc} (\Omega_\theta)$
--- in fact, it belongs to $W^{1,2}_\mathrm{loc}$ --- the first two
contributions to the integral are finite. Combining this with the
fact that $Q^\theta_m[\tilde\psi^{(j)}] < \infty$ we infer that
also the last contribution to the integral is finite, so the
derivative exists. Using then the integrand continuity with
respect to $\theta$ and the dominated convergence theorem, we
arrive at the result.
\end{proof}

Note that the terms appearing in the last expression have
different signs so it is not apriori clear whether the eigenvalue
behave monotonously as the angle changes. Nevertheless, the
geometrically induced binding becomes stronger as $\theta$
increases, in particular, we have the following result.

\begin{proposition} \label{lowerb}
Set $\lambda_0:= \pi^{-2} j_{0,1}^2\approx 0.58596$ where
$j_{0,1}$ is the first zero of Bessel function $J_0$. To any
$N\in\N$ and $\bar\lambda\in (\lambda_0,1)$ there is a $\theta_0$
such that $H^\theta$ has at least $N$ eigenvalues below
$\bar\lambda$ for any $\theta\in (\theta_0, \frac12\pi)$.
\end{proposition}
\begin{proof}
By Dirichlet bracketing \cite[Sec.~XIII.15]{RS78} it is sufficient
to find a region $\Sigma\subset \Sigma_\theta$ such that
$-\Delta_\mathrm{D}^{\Sigma}$ has this property. Consider a cylinder of radius $R<\pi$ and length $L$. The eigenvalues of the corresponding Dirichlet Laplacian are easily found; it yields the condition
 $$
 \lambda_0 \left( \frac{\pi}{R}\right)^2 + \left( \frac{\pi
 N}{L}\right)^2 < \bar\lambda\,.
 $$
We embed the cylinder axially into $\Sigma_\theta$ with one lid touching the tip of the inner surface; when $\theta$ becomes close enough to $\frac12\pi$ we can chose $R$ close to $\pi$ and $L$ large so that the above condition is satisfied.
\end{proof}



\section{Numerical results}\label{s: numerical}

In order to get a more detailed information and a better insight let us study next the spectrum numerically. As a consequence of Proposition~\ref{s-states}, we may limit ourselves to two--dimensional geometry, \emph{i.e.} to consider the spectrum of (\ref{partial2}) for $m=0$. Recall that the corresponding eigenfunctions are independent of the angular variable $\phi$ so we can choose $\phi=0$ arriving at the two--dimensional problem depicted in Fig.~\ref{fig1}. Note also that the alternative formulation (\ref{partial4}) leads to the same results.

The first difficulty we meet is the absence of reliable estimates
concerning the span of eigenfunctions. The knowledge of the region
where the eigenfunctions are negligible is important. We have to
carry out integration over an infinite region and we need to cut
the integration somewhere, unless we transform the problem to a
finite region. This, however, brings problems of other kinds
ruling out, in particular, use of the mode-matching approach.

From general considerations we know that the smaller the gap between the eigenvalue and the continuum threshold, of course, the farther is the wavefunction stretched. Also, for small opening angles eigenfunctions spread farther. This rule of thumb may be of help as a rough estimate of integration region.

\subsection{The method used} \label{ss: method}

We must admit that we have not found a method which would be
efficient enough and intuitive at the same time. To summarize our
experience we prefer finite element methods. In view of the
absence of knowledge of a practical span of configuration space
mentioned above even this approach requires separate treating of
respective levels. In order to minimize the number of elements we
took further advantage of symmetry and limited the problem to the
region shown in the left part of Fig.~\ref{fig2} with the
appropriate boundary conditions. The sought eigenvalues of the
conical layer are then eigenvalues of the operator
(\ref{partial2}) with $m=0$ on this region.

We used 2nd order elements. The total number of degrees of freedom
ranged from 50000 to about 300000 depending on eigenfunction
extent. All runs for small $\theta$ (most demanding) took less
than one minute each on an Intel Core 2 (1.86 GHz).

We have also tried to employ the Galerkin method. The idea was to
choose a set of test functions satisfying the boundary conditions.
This, however, leaves still leeway in representing the subspace to
which the problem is projected, \emph{i.e.} how to choose the
particular form of the test functions. The choices we attempted
led to substantially more time consuming calculation than FEM, so
fitting test functions are to be found yet.

\subsection{The results} \label{ss: results}

\begin{figure}
   \begin{center}
     \includegraphics{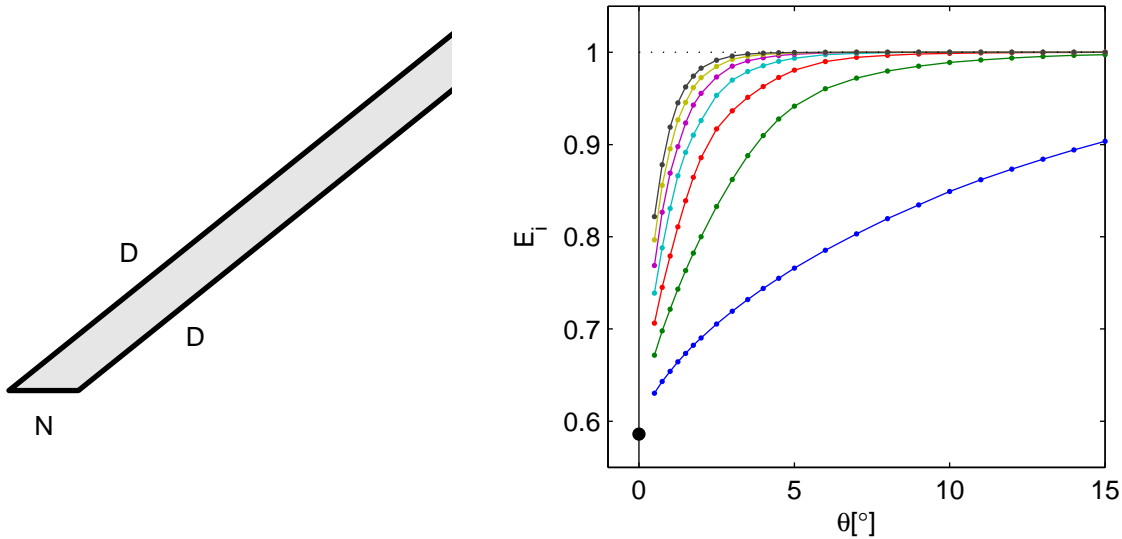}
     \caption{The left figure depicts the reduced 2D region on which the solution is sought. The right one shows the dependence of eigenvalues
     on the opening angle; the first seven of them are shown (colour
     online). The essential spectrum threshold $E=1$ is the accumulation
     point of all of them and they approach it from below as
     $\theta\to \frac12\pi-$. The thick dot marks the value $\lambda_0$ indicated in Proposition~\ref{lowerb}; the picture suggests that the eigenvalues tend to it monotonously as $\theta\to 0+$.
     }\label{fig2}
   \end{center}
\end{figure}

The first question concerns the dependence of eigenvalues on the opening angle illustrated on Fig.~\ref{fig2} in the interval between 1 and 15 degrees. It agrees well with the general picture: the planar layer, $2\theta=\pi$, does not support any eigenvalue below the threshold, while for any $\theta<\pi/2$ there are by Theorem~\ref{disc_spec} infinitely many such isolated eigenvalues, albeit all are weakly bound for sufficiently open conical layers. As the opening angle becomes more acute the low-lying states become gradually more strongly bound. The numerical results show no degeneracies.

\begin{figure}
   \begin{center}
     \includegraphics[width=27em]{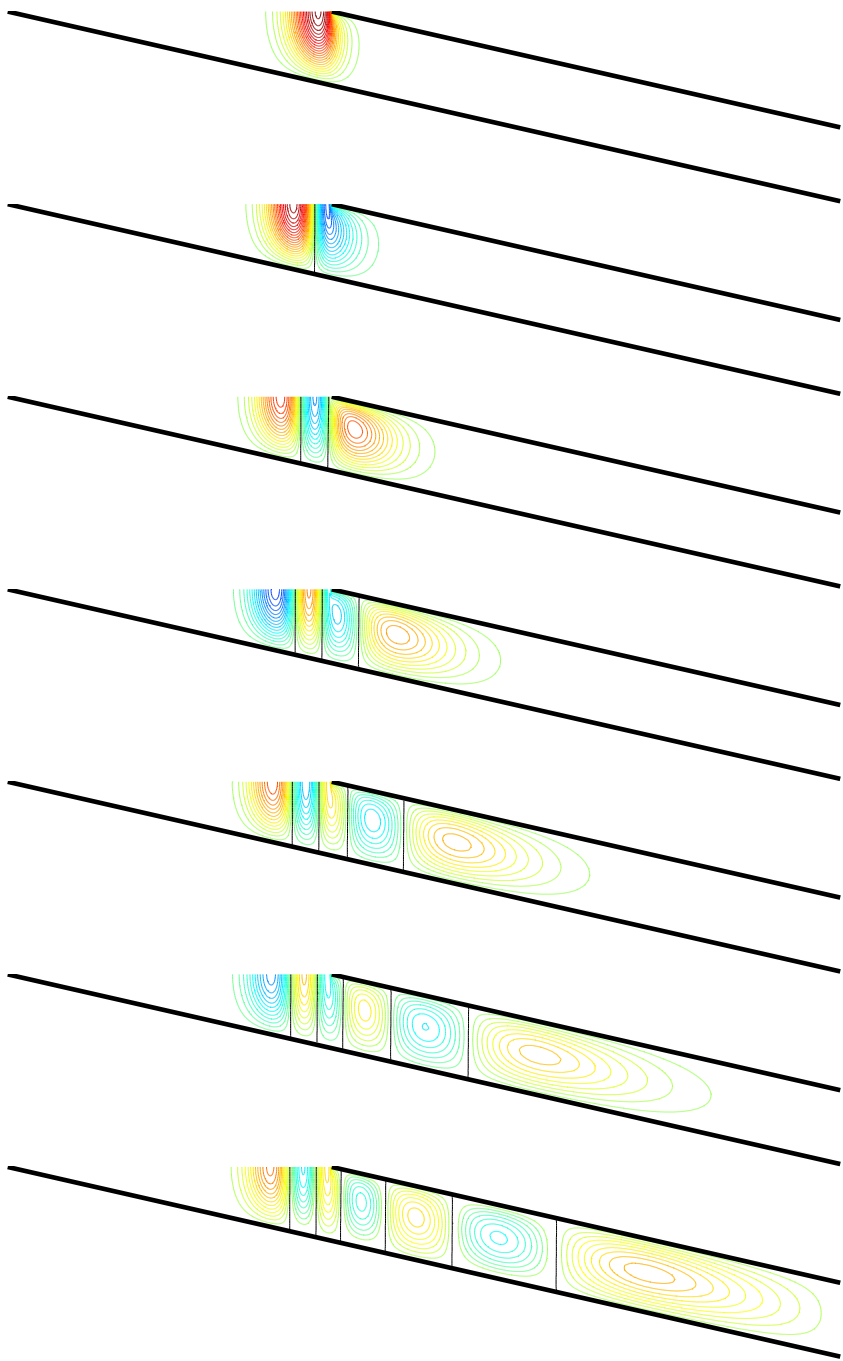}
     \caption{The contour plot of the first seven eigenfunctions
     (colour online) for $\theta=2.5^\mathrm{o}$. In view of the
     axial symmetry, we show only one half of $\psi(r,z)$. The the $r-$ and $z-$scales are different, the vertical axis being scaled by factor five: the transversal width is $\pi$ and $z$ ranges from $0$ to $225$. Nodal
     lines are indicated in black colour; higher eigenfuctions
     have clearly a quasi-onedimensional character.}\label{fig3}
   \end{center}
\end{figure}

As we move away from the ground state, structure of the
eigenfunctions acquires more and more a one-dimensional character. This is demonstrated in Fig~\ref{fig3} when the seven lowest eigenfunctions for the angle $\theta=2.5^\mathrm{o}$ are calculated. As expected, the ground-state eigenfunction can be chosen positive. On the other hand, the excited states display a number of nodal lines (nodal surfaces of the full problem) that is in a one-to-one correspondence with the eigenvalue index of the
sequence $\{\lambda_i(\theta)\}$, in an analogy with the classical oscillation theorem. Notice also that the node distances are increasing with the distance of the cone tip. This behaviour has to be associated with the fact that the effective-curvature induced potential in (\ref{partial4}) is proportional to $s^{-2}$ having thus the same scaling behaviour as the kinetic part of the Hamiltonian.

\begin{figure}
   \begin{center}
     \includegraphics[width=27em]{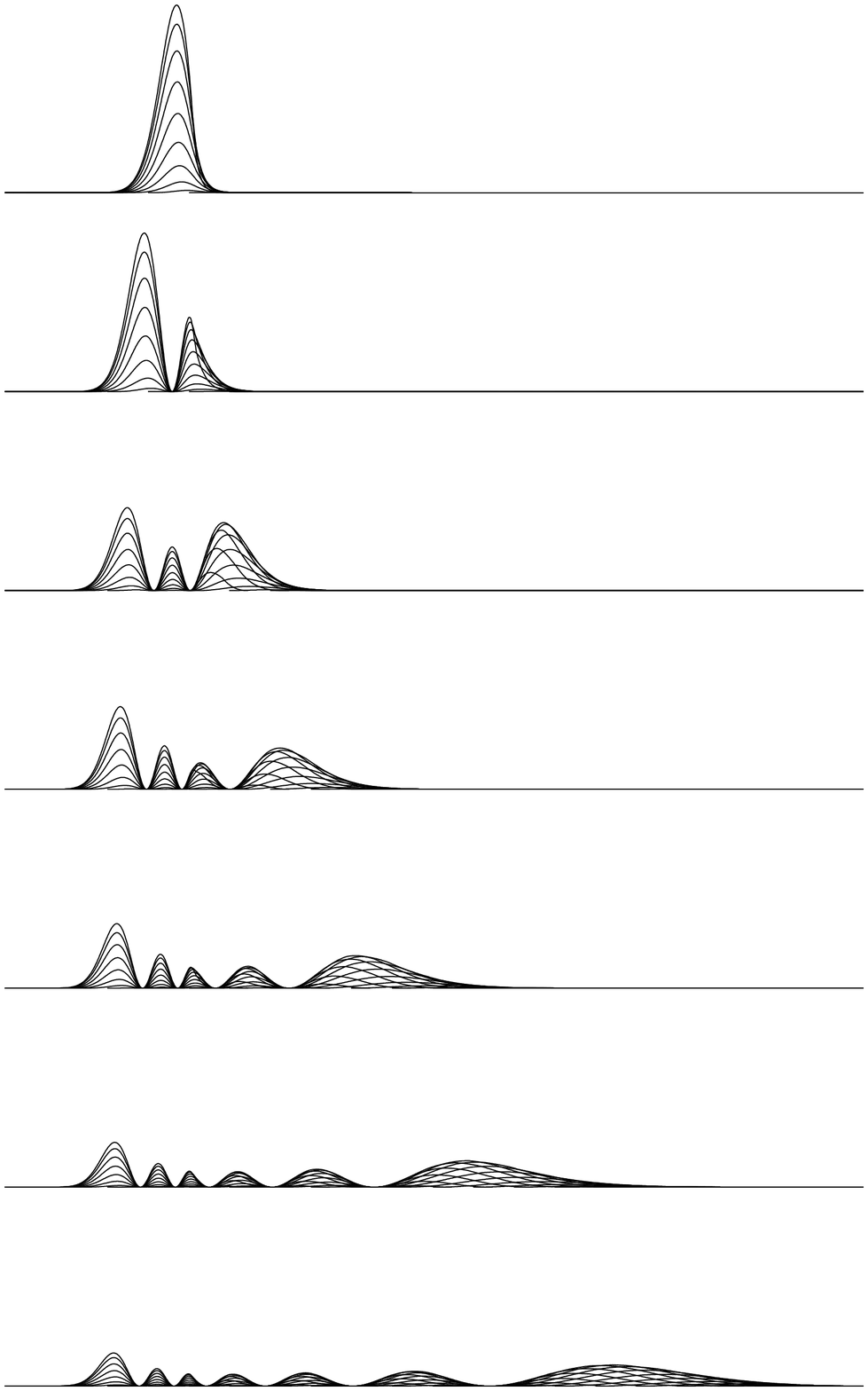}
     \caption{Side view (\emph{i.e.}, with zero elevation angle) of
$|\psi|^2$ for the same situation as in Fig. 3. All the
eigenfunctions are normalized to the same constant, so that
relative height ratios reflect the true situation (not distorted
by different normalization of individual functions). There are two
areas with higher probability, one in the vicinity of the inner
tip and the other one at the ``far end'' of the function extent.
The latter moves away from the vertex as the eigenvalue number
increases.}\label{fig4}
   \end{center}
\end{figure}

Another observation worth commenting is shown on Fig.~\ref{fig4},
where we have plotted (a side view of the) squared eigenfunctions
$|\psi^i|^2$ with $\,i=1,\ldots,7\,$ for the same opening angle
$\theta$ as above. As the index $i$ grows, dominant regions lie
around the inner tip and at the ``far end'' of eigenfunction
extent, more exactly, between the last two nodal lines.
Consequently, location of the regions with the largest probability
of finding the particle in the conical layer depends on the
eigenvalue number in this way.

\section{Concluding remarks}\label{s: concluding}

Our discussion has not exhausted by far all open questions
concerning spectra of conical layers. We have also mentioned the
eigenvalue monotonicity which is expected to hold and supported by
the above numerical result. A related question concerns the
simplicity of the spectrum. It is again seen in the numerical
example and its validity is strongly supported by the
quasi-one-dimensional character of the higher eigenfunctions, but
these considerations cannot replace a rigorous proof.

Other open questions are associated with the behaviour for extreme
values of $\theta$ which we left out in the previous section
because these cases are numerically demanding.  For small values
of $\theta$ the density of eigenvalues above $\lambda_0\approx
0.58596$ should grow in accordance with Proposition~\ref{lowerb}.
What we see nevertheless is that already for
$\theta=2.5^\mathrm{o}$ the excited states have nodal lines (or
nodal surfaces of the full problem) inside the cone ``cap''; one
can ask about critical values of $\theta$ for which one of those
touches the inner tip. On the other hand, one can ask about the
spectral asymptotics as $\theta\to \frac12\pi-$. In contrast to
\cite{EK01} the geometric perturbation is not compactly supported
for conical layers and the effective potential is on the
borderline between short- and long-range cases, so it is not a
priori clear what the asymptotic expansion could be.

Of interest is also the continuous spectrum of conical layers,
both from the point of view of scattering --- existence of wave
operators and their asymptotic completeness --- as well as of
resonances. Being bound by the deadline of this issue, we leave
all these questions to a later investigation.

\section*{Acknowledgments}
We dedicate this paper to the memory of Pierre Duclos, a longtime friend and collaborator, who was always interested in geometrically induced spectral effects. We thank the referees for the remarks which helped us to improve the text. The research was supported by the Czech Ministry of Education,
Youth and Sports within the project LC06002.

\end{document}